\documentclass[9pt,shortpaper,twoside,web]{ieeecolor}
\usepackage{generic}
\usepackage{cite}
\usepackage{amsmath,amssymb,amsfonts}
\usepackage{algorithm}
\usepackage{algorithmic}
\usepackage{graphicx}
\usepackage{textcomp}
\usepackage{graphicx} 
\usepackage{url}
\usepackage{verbatim}
\usepackage{stmaryrd}  
\usepackage{color}
\usepackage{bbm}
\usepackage{mathrsfs}
\usepackage{caption}
 
\usepackage{amsthm}
\usepackage{nomencl}
\usepackage{multirow}
\usepackage{chngpage}
\usepackage{array}
\usepackage{comment}
\usepackage{enumerate}   
\usepackage{transparent}
\usepackage{mathtools}
\usepackage{datatool} 
\usepackage[normalem]{ulem}
\usepackage{longtable,array}

\makenomenclature
\allowdisplaybreaks

\theoremstyle{plain}
\newtheorem{thm}{\bf Theorem} 
\newtheorem{lem} {\bf Lemma} 
\newtheorem{cor} {\bf Corollary}
\newtheorem{pro} {\bf Proposition}

\theoremstyle{definition}
\newtheorem{rem} {\bf Remark}
\newtheorem{defn} {\bf Definition}
\newtheorem{example} {\bf Example}

\newcommand\bx{\mathbf x} 
\newcommand\bz{\mathbf z} 
\newcommand\rk{{\rm rank\,}} 
 
\newcommand\rd{{\rm d}} 
\newcommand{\setdef}[2]{\left\{ #1 \left|\ \vphantom{#1} #2 \right.\right\}}

\newcommand{\sspan}[1]{{\rm span}_{\mathcal K}\left\{ #1 \right\}}
\newcommand{\kspan}[1]{{\rm span}_{\mathcal K(\delta]}\left\{ #1 \right\}}
\newcommand{\smatrix}[1]{\left[ \begin{smallmatrix}  #1 \end{smallmatrix}\right]}

\def\BibTeX{{\rm B\kern-.05em{\sc i\kern-.025em b}\kern-.08em
    T\kern-.1667em\lower.7ex\hbox{E}\kern-.125emX}}
\begin{document}
\title{Implicit function theorem for nonlinear time-delay systems with algebraic constraints}
\author{Yahao Chen,  Malek Ghanes \IEEEmembership{Member, IEEE}, and Jean-Pierre Barbot, \IEEEmembership{Senior Member, IEEE}
\thanks{Yahao Chen is with Nantes University, Centrale Nantes, LS2N UMR CNRS 6004, France (e-mail: yahao.chen@ls2n.fr). }
\thanks{Malek Ghanes is with Nantes University,  Centrale Nantes, LS2N UMR CNRS 6004, France (e-mail: malek.ghanes@ls2n.fr).}
\thanks{Jean-Pierre Barbot is with 
ENSEA, Quartz EA 7393 and  Nantes University, Centrale Nantes, LS2N UMR CNRS 6004, France  (e-mail: barbot@ensea.fr).}}

\maketitle

\begin{abstract}
 In this note, we discuss a  generalization of the well-known implicit function theorem to the time-delay case. We show that the latter  problem is closely related to the bicausal changes of coordinates  of time-delay systems \cite{califano2014coordinates,califano2016accessibility}. An iterative algorithm is proposed to check the conditions and to construct the desired bicausal change of coordinates for the proposed implicit function theorem. Moreover, we show that   our results can be applied to   delayed differential-algebraic equations (DDAEs) to
  reduce their indices and   to get their solutions. Some numerical examples are given to illustrate our results. 
\end{abstract}

\begin{IEEEkeywords}
nonlinear systems, time-delay, bicausal changes of coordinates, implicit function theorem, causality, differential-algebraic equations 
\end{IEEEkeywords}

\section{Introduction}\label{sec:1}
We start from three different equations with time-delay variables:
$$
\begin{aligned}
a(\bx_1,\bx_2)&=x_1(t)x_2(t-1)+x_2(t)x_2(t-1)+e_1=0,\\
b(\bx_1,\bx_2)&=x_1(t)x_2(t-1)+x_1(t-1)x_2(t)x_2(t-2)+e_2=0,\\
c(\bx_1,\bx_2)&=x_1(t)x_1(t-1)+x_2(t)x_2(t-1)+e_3=0,\\
\end{aligned}
$$
where $(\bx_1,\bx_2)=(x_1(t),x_2(t),x_1(t-1),x_2(t-1),x_2(t-2))$ and $e_1,e_2,e_3$ are nonzero constants. The purpose is to express $x_1(t)$ as a function of $x_2(t)$ and its time-delays from each algebraic equation. For instance, it is clear to get $x_1(t)=\frac{-e_1-x_2(t)x_2(t-1)}{x_2(t-1)}$ for $x_2(t-1)\neq 0$ from the first equation, while it is not obvious  if we can have similar conclusions for the other two equations. In the delay-free case, given some algebraic equations $\lambda(x_1,x_2)=0$, where $\lambda\in \mathcal K^p$, $x_1\in \mathbb R^p$, $x_2\in \mathbb R^{n-p}$ and $\mathcal K$ denotes the field of meromorphic functions, if the matrix $\frac{\partial\lambda}{\partial x_1}(x_1,x_2)\in \mathcal K^{p\times p}$ is invertible for all $(x_1,x_2)\in \mathbb R^n$ such that $\lambda(x_1,x_2)=0$ (or, 
a simpler but stronger condition, for all $(x_1,x_2)\in \mathbb R^n$), then by the classical implicit function theorem (see e.g. \cite{krantz2002implicit}), there exist functions $\gamma:\mathbb R^{n-p}\to \mathcal K^p $ such that $\lambda(x_1,x_2)=0$ implies $x_1=\gamma(x_2)$. We will study in this note a   generalization of the implicit function theorem to algebraic equations with time-delays.

To deal with functions with time-delay variables, the algebraic framework proposed in \cite{xia2002analysis} is a very useful tool. There are many applications of  this framework see e.g.,   \cite{marquez2007new,marquez2002observability,halas2013retarded,zheng2015unknown, bejarano2021zero} for the problems like  observations and  structure analysis for time-delay systems, and more recently,   the  papers \cite{califano2013controllability,califano2014coordinates,kaldmae2015integrability,califano2016accessibility,califano2020observability} and the book \cite{califano2021nonlinear} for  extensions of the classical geometric control methods to nonlinear time-delay  systems.

With the help of the algebraic framework, we will show in section~\ref{sec:2}   that although it is not possible to express $x_1(t)$ as a function of $x_2(t), x_2(t-1),x_2(t-2)$  from the last two equations, by a bicausal change of coordinates (see Definition~\ref{Def:bicausal} below) $(\tilde x_1,\tilde x_2)=\varphi(\bx_1,\bx_2)$, the equation $b(\varphi^{-1}(\tilde \bx_1,\tilde \bx_2))=0$, i.e., $b(x_1,x_2)=0$ in $(\tilde \bx_1,\tilde \bx_2)$-coordinates,  implies  $\tilde x_1=\gamma(\tilde \bx_2)$ for some function $\gamma$ (but we can not find such a bicausal  change of coordinates for  $c(\bx)=0$).
The first problem  is that   when is it possible and how do we find such a bicausal coordinates transformation for a given time-delay algebraic equation?  It turns out that such a problem is closely related to   when   functions  with time-delay variables can be regarded as new bicausal coordinates and how to construct their complementary bicausal coordinates, the latter problems are discussed in \cite{marquez2007new,califano2014coordinates,califano2016accessibility}. We will recall some results from \cite{califano2016accessibility} and add two extra equivalent conditions as Theorem~\ref{Thm:impl}   in order to explain the  relations. Then a generalization of implicit function theorem to time-delay equations is given  as a corollary of Theorem~\ref{Thm:impl}. 

To check the equivalent conditions of Theorem \ref{Thm:impl}, we need to either construct the right-annihilator/kernel or the right-inverse  of  some polynomial matrix-valued functions, which can be done with the help of Smith canonical form for polynomial matrix-valued functions  (see \cite{marquez2007new}, also \cite{gantmacher1959matrix} for polynomial  matrices with entries in $\mathbb R[\delta]$). We will discuss in section \ref{sec:4} below that the latter method has troubles when   checking the necessity of those conditions. To deal with the latter problem, we propose an iterative algorithm  by reducing  the polynomial degree of the polynomial matrix-valued functions via bicausal changes of coordinates, which eventually allows  to check the conditions of Theorem \ref{Thm:impl} and to construct the desired complementary bicausal  coordinates.

Another contribution of this note is to apply the proposed implicit function theorem to delayed differential-algebraic equations (DDAEs), i.e., implicit time-delay equations (see e.g. \cite{campbell19912,du2013stability,ha2016analysis,trenn2019delay} for linear DDAEs and \cite{venkatasubramanian1994time,ascher1995numerical,zhu1998asymptotic} for nonlinear DDAEs). It is well-known that for delay-free differential-algebraic equations (DAEs), the classical implicit function theorem is an essential tool for its index-reduction problem, e.g., given a semi-explicit DAE $\dot x_1=f(x_1,x_2), 0=g(x_1,x_2)$, if $\frac{\partial g(x_1,x_2)}{\partial x_2}\neq 0$ for all $(x_1,x_2)\in \mathbb R^2$ (i.e., the DAE is index-1), then to reformulate the DAE as an ordinary differential equation (ODE), we use the implicit function theorem to get $x_2=\eta(x_1)$  from the algebraic constraint for some functions $\eta$,  then it results in  an ODE $\dot x_1=f(x_1,\eta(x_1))$. For a   high-index DAE, the geometric reduction method can be used to reduce the index, see e.g. \cite{RabiRhei02,Riaz08,chen2022geometric,chen2021normal}.  We will show below that by assuming that 
the algebraic constraints of DDAEs satisfy the conditions of the proposed implicit function theorem, a time-delay version of the  geometric reduction method can be realized.

This note is organised as follows. Notations and the definitions of some notions in the algebraic framework are   given in section~\ref{sec:2}. The time-delay implicit function theorem is discussed in section \ref{sec:3}. The algorithm to check the conditions of the time-delay implicit function theorem is given in section \ref{sec:4}. In section \ref{sec:5}, we discuss the index reduction algorithm and the solutions of nonlinear DDAEs by applying the results of sections \ref{sec:3} and \ref{sec:4}. The conclusions and perspectives are put into section~\ref{sec:6}.
\section{Notations and Preliminaries}\label{sec:2}
We  follow  the algebraic framework of time-delay systems proposed in \cite{xia2002analysis}, the notations below are taken from those in e.g., \cite{xia2002analysis,califano2016accessibility,califano2021nonlinear}. In this note, we do not deal with singularities and assume throughout that $f(\bx)\equiv 0$ holds for no non-trivial meromorphic function $f$.
\begin{small} 
	\begin{tabular}{p{0.13\columnwidth}p{0.8\columnwidth}}
	${\rm I}_r$   & identity matrix of $\mathbb R^{r\times r}$.\\
	$x(\pm j)$ & $x(t\pm j)$, $j\ge 0$.\\
	$\bx_{[\underline j,\bar j]}$& $[x^T(-\underline j),\ldots, x^T(-\bar j)]^T\in \mathbb R^{(\bar j-\underline j+1)n}$, $\underline j\le \bar j$. Notice that both $\underline j$ and $\bar j$ can be negative indicating the forward time-shifts of $x(t)$.\\
	$\bx_{[\bar j]}$ &  $\bx=\bx_{[0,\bar j]}=[x^T,x^T(-1),\ldots, x^T(-\bar j)]^T\in \mathbb R^{(\bar j+1)n}$, where $x=x(0)=x(t)\in \mathbb R^n$ and $\bar j\ge 0$.\\
	$\mathcal K$ & the field of meromorphic functions.\\
	$\rd $ & the differential operator: for $\xi(\bx_{[\bar j]})\in \mathcal K$ and $\lambda(\bx_{[\bar j]})\in \mathcal K^p$, $\rd \xi(\bx_{[\bar j]})=\sum\limits^{\bar j}_{j=0}\frac{\partial \xi(\bx_{[\bar j]})}{\partial x(-j)}\rd x(-j)$ and $\rd \lambda=\smatrix{\rd \lambda_1\\ \cdots \\ \rd \lambda_p}$.\\
	$\delta$ & the backward time-shift  operator: for $a(t),\xi(t)\in \mathcal K$, $\delta^{j}\xi(t)=\xi(-j)$        and       $ \delta^j(a(t)\rd \xi(t))=a(-j) \rd \xi(-j)$. \\
 $\Delta$ & the forward time-shift  operator: for $a(t),\xi(t)\in \mathcal K$,  $\Delta^{j}\xi(t)=\xi(+j)$        and       $ \Delta^j(a(t)\rd \xi(t))=a(+j) \rd \xi(+j)$. \\
 $\mathcal K(\delta]$ &the left (Ore-)ring of polynomials in $\delta$ with entries in $\mathcal K$, any $\alpha(\bx,\delta)\in \mathcal K(\delta]$ has the form $\alpha(\bx,\delta)=\sum\limits^{\bar j}_{j=0}\alpha^j(\bx)\delta^j$, where $\alpha^j(\bx)\in \mathcal K$.\\  
  $\deg(\cdot) $& The  polynomial degree. For $\alpha(\bx,\delta)\in \mathcal K(\delta]$ above, $\deg(\alpha)=\bar j$. For $\beta(\bx,\delta)=[\beta_1(\bx,\delta), \ldots,\beta_n(\bx,\delta)]\in \mathcal K^n(\delta]$, $\deg(\beta)=\max\{\deg(\beta_i), 1\le i\le n\}$.\\
  $\wedge$ & exterior product \\
	\end{tabular} 
\end{small} 

 The sums and multiplications for any two elements of $\mathcal K(\delta]$ are well-defined \cite{xia2002analysis},  so  is the rank of a matrix $A(\cdot,\delta)\in \mathcal K^{r\times m}(\delta]$  over $\mathcal K(\delta]$, denoted by $\rk_{\mathcal K(\delta]} A(\cdot,\delta)$. 
Remark that a polynomial matrix-valued function $A(\cdot,\delta)\in \mathcal K^{r\times r}(\delta]$ is of full rank over $\mathcal K(\delta]$   does not necessarily mean that $A(\cdot,\delta)$ has a polynomial inverse over $\mathcal K(\delta]$, the following notion of unimodularity  generalizes that of invertibility of non-polynomial matrices.
 \begin{defn}[\cite{xia2002analysis,marquez2007new}]
A matrix $A(\cdot,\delta)\in \mathcal K^{r\times r}(\delta]$ is called  \emph{unimodular} if there exists a matrix $B(\cdot,\delta)\in \mathcal K^{r\times r}(\delta]$ such that $A(\cdot,\delta)B(\cdot,\delta)=B(\cdot,\delta)A(\cdot,\delta)={\rm I}_{r}$.
 \end{defn} 
Denote the vector space generated by the differentials $\rd x(-j)$, $j\ge 0$ over $\mathcal K$ by $\mathcal E$. An element $\omega \in \mathcal E$ is called one-form. The one-form $\omega$ is \emph{exact}, i.e., there exists $\lambda\in \mathcal K$ such that $\omega=\rd \lambda$, if and only if $\rd \omega=0$ (Poincar\'e  lemma \cite{lee2001introduction}). A $p$-dimensional codistribution $\sspan{\omega_1, \ldots,\omega_p}$ is \emph{integrable}, i.e., there exist $\lambda_1,\ldots,\lambda_p\in \mathcal K$ such that $\sspan{\omega_1, \ldots,\omega_p}=\sspan{\rd\lambda_1, \ldots,\rd\lambda_p}$, if and only if $\rd \omega_i\wedge \omega_1\wedge \cdots\wedge \omega_p=0$, for $1\le i\le p$ (Frobenius theorem   \cite{lee2001introduction}).   The sets of one-forms defined over the ring $\mathcal K(\delta]$ have both the structure of a vector space $\mathcal E$ over $\mathcal K$ and the structure of a (left)-module, $\mathcal M=\kspan{\rd x}$. A (left)-submodule of $\mathcal M$ consists of all possible linear combinations of given one forms over the ring $\mathcal K(\delta]$. Denote $\mathcal O:=\kspan{\omega_1,\ldots,\omega_p}\subseteq \mathcal M$
the submodule generated by one forms $\omega_1,\ldots,\omega_p$ over $\mathcal K(\delta]$. The \emph{right-annihilator} (or the kernel) of the submodule $\mathcal O$ is spanned by all vectors $\tau(\cdot,\delta)\in \mathcal K^n(\delta]$ such that $\omega_i(\cdot,\delta)\tau(\cdot,\delta)=0$ for $1\le i\le p$. The right-annihilator is called \emph{causal} if there are no forward time-shifts on the variables of the coefficients of $\tau(\cdot,\delta)$. The closure of the submodules of $\mathcal M$  recalled below will play an important role.
 \begin{defn}[\cite{xia2002analysis}]
 	Given a finite generated module $\mathcal M$,   let $\mathcal N$ be a submodule  of  $\mathcal M$ of dimension $r$   over $\mathcal K(\delta]$, the closure of $\mathcal N$ is the submodule
 	$$
 	\overline {\mathcal N}:=\setdef{\omega\in \mathcal M}{\exists\,  0\neq \alpha(\cdot,\delta)\in \mathcal K(\delta],\  \alpha(\cdot,\delta)\omega\in \mathcal N},
 	$$
 	or equivalently, $\overline {\mathcal N}$ is the largest submodule  of $\mathcal M$ which contains $\mathcal N$ and is of rank $r$. The submodule $\mathcal N$ is called \emph{closed} if ${\mathcal N}=\overline {\mathcal N}$.
 \end{defn}
 The following definition of bicausal change of coordinates will be used in the note:
\begin{defn}[bicausal  coordinates changes \cite{marquez2002observability,califano2016accessibility}]\label{Def:bicausal}
Consider a system (differential or not) with state coordinates $x\in \mathbb R^n$. A mapping $z=\varphi(\bx_{[\bar j]})\in \mathcal K^n$, is called a bicausal change of coordinates if there exist an integer $\bar j_z\ge 0$ and  a mapping $\varphi^{-1}\in \mathcal K^n$ such that $x=\varphi^{-1}( \bz_{[\bar j_z]})$.
\end{defn}
Remark that a mapping $z=\varphi(\bx)$ is a bicausal change of coordinates if and only if $T(\bx,\delta)\in \mathcal K^{n\times n}(\delta]$ is a unimodular matrix \cite{califano2021nonlinear}, where  $\rd z=\rd\varphi(\bx)=T(\bx,\delta)\rd x $. For a function $\lambda(\bx)\in \mathcal K$, we will simply write  $\lambda$ in $z$-coordinates as $$\lambda(\bz):=\lambda(\varphi^{-1}(\bz),\ldots, \varphi^{-1}(\bz(-\bar j))).$$
\section{implicit function theorem for time-delay algebraic equations}\label{sec:3}
Now consider the time-delay algebraic equations $\lambda(\bx)=\lambda(\bx_1,\bx_2)=0$ with $\lambda\in \mathcal K^p$.  The differentials of $\lambda$ satisfy $\rd\lambda(\bx)=T_1(\bx,\delta)\rd x_1+T_2(\bx,\delta)\rd x_2 =0$. If $T_1(\bx,\delta)\in \mathcal K^{p\times p}(\delta]$ is unimodular, then   $\rd x_1=-T^{-1}_1T_2(\bx,\delta)\rd x_2$. Thus by Poincar\'e lemma, there exist functions $\gamma\in \mathcal K^p$ such that $x_1=\gamma(\bx_2)$. The last analysis  explains why we can get $x_1$ as a function of $x_2$ and $x_2(-1)$ from $a(\bx)=0$ in section \ref{sec:1}, clearly, $\rd a=x_2(-1)\rd x_1+(x_1\delta+x_2\delta+x_2(-1))\rd x_2$ and $x_2(-1)$ is unimodular.  To have a similar result for $b(\bx)=0$, we have to use bicausal changes of coordinates as shown in  Example \ref{Ex:alg1} below.
In general, we have the following theorem, in which items (i) and (ii) are taken from Theorem~2 of \cite{califano2016accessibility}, items (iii) and (iv) are new and serve to our problem. 

We    use the following condition  \textbf{(C)} for any   submodule $\mathcal N\subseteq \mathcal M$: 

 \textbf{(C)}:  $\mathcal N$ is closed and its right-annihilator is causal.

\begin{thm}\label{Thm:impl}
Consider  $p$-functions $\lambda_k(\bx)\in \mathcal K $, $1\le k\le p$, of the variables $x\in \mathbb R^n$ and its time-delays. Define the submodule $\mathcal L:=\kspan{  \rd \lambda_k(\bx), 1\le k \le p}$ and assume that $\dim \mathcal L=p$ over $\mathcal K(\delta]$. Then the following statements are equivalent:
\begin{itemize}
\item [(i)]  $\mathcal L$ satisfies \textbf{(C)}.
\item [(ii)] There exist $n-p$ functions $\theta_1(\bx),\ldots,\theta_{n-p}(\bx)$ such that $\kspan{\rd\lambda_1,\ldots,\rd \lambda_{p},\rd \theta_1,\ldots, \rd \theta_{n-p}}=\kspan{\rd x}$, i.e., $\tilde x=[\lambda_1(\bx),\ldots, \lambda_{p}(\bx), \theta_1(\bx),\ldots, \theta_{n-p}(\bx)]^T$ is a bicausal change of coordinates. 

\item [(iii)]   $L(\bx,\delta)\in \mathcal K^{p\times n}(\delta]$, where $\rd \lambda(\bx)=L(\bx,\delta)\rd x$ and $\lambda=  [\lambda_1,\ldots,\lambda_p]^T$, has a polynomial right-inverse, i.e., $\exists L^{\dagger}(\bx,\delta)\in \mathcal K^{n\times p}(\delta]$ such that $LL^{\dagger}={\rm I}_{p}$.

    \item [(iv)] There exists a bicausal change of coordinates $\smatrix{\tilde x_1\\
\tilde x_2}
 =\varphi(\bx)$ with $\tilde x_1\in \mathbb R^{p}$ and $\tilde x_2\in \mathbb R^{n-p}$ such that $L_1(\cdot,\delta)$ is unimodular and $L_2(\cdot,\delta)\not\equiv 0$, where $L_1(\tilde \bx_1,\tilde \bx_2,\delta)\rd \tilde x_1+L_2(\tilde \bx_1,\tilde \bx_2,\delta)\rd \tilde x_2=\rd \lambda(\tilde \bx_1,\tilde\bx_2)$.
\end{itemize}
\end{thm} 
\begin{rem}\label{Rem:1}
 The added condition (iii) is  in some cases easier to be checked than condition (i) because the right-annihilator could be rendered causal even some non-causal terms shows up in the initial calculation of the kernel. Take $\lambda=x_1(-1)x_2+x_2^2$ from Example~3.6 of \cite{califano2016accessibility}, in which it is claimed that the right-annihilator of $\mathcal L=\kspan{x_2\delta\rd x_1+(x_1(-1)+2x_2)\rd x_2}$ is \emph{not} causal because the non-causal  functions $r(\bx_{[-1,0]},\delta)=\smatrix{-2x_2(+1)-x_1\\x_2\delta}$ belongs to the   right-annihilator. However,  $L(\bx,\delta)=\smatrix{x_1\delta,& x_1(-1)+2x_2}$ has a polynomial right-inverse $$L^{\dagger}(\bx,\delta)=\smatrix{0\\ \frac{1}{x_1(-1)+2x_2}}$$ (for $x_1(-1)+2x_2\neq 0$). In fact, $r(\bx,\delta)$ can be rendered as $\smatrix{-1\\\frac{x_2}{x_1(-1)+2x_2}\delta}$ proving that the right-annihilator is actually causal. Indeed, choose $\theta=x_1$, we have $\smatrix{\lambda(\bx)\\ \theta(\bx)}$ is a bicausal change of coordinates since $\smatrix{\rd\lambda(\bx)\\ \rd\theta(\bx)} =\Theta(\bx,\delta)\smatrix{\rd x_1\\ \rd x_2}$ and $\Theta(\bx,\delta)=\smatrix{x_2\delta& x_1(-1)+2x_2\\
 1&0}$ is unimodular, hence by the equivalence of items (i) and (ii), the right-annihilator of $\mathcal L$ is indeed causal.
 \end{rem}
\begin{proof}
The proof of $(i)\Leftrightarrow (ii)$ can be found in \cite{califano2016accessibility} and \cite{marquez2007new}.

  (i)$\Rightarrow$(iii): Assume that item (i) holds, then by Lemma 12 of \cite{marquez2007new}, there exist  two (causal) unimodular matrices $P(\bx,\delta)\in \mathcal K^{p\times p}(\delta]$ and $Q(\bx,\delta)\in  \mathcal K^{n\times n}(\delta]$ such that   $P(\bx,\delta)L(\bx,\delta)Q(\bx,\delta)=[{\rm I}_p \ 0]$. It follows that  $L(\bx,\delta)Q(\bx,\delta)=[P^{-1}(\bx,\delta) \ 0]$ and thus $L(\bx,\delta)Q_1(\bx,\delta)=P^{-1}(\bx,\delta)$, where $Q=\smatrix{Q_1&Q_2}$ and $Q_1(\bx,\delta)\in \mathcal K^{n\times p}(\delta]$. Hence $L(\bx,\delta)Q_1(\bx,\delta)P(\bx,\delta)={\rm I}_p$ and $L^{\dagger}(\bx,\delta)=Q_1(\bx,\delta)P(\bx,\delta)$ is a polynomial right-inverse of $L(\bx,\delta)$.

(iii)$\Rightarrow$ (i): Assume that there exists $L^{\dagger}(\bx,\delta)\in \mathcal K^{n\times p}(\delta]$ such that $L(\bx,\delta)L^{\dagger}(\bx,\delta)={\rm I}_p$. Then by Lemma 4 of \cite{marquez2007new}, there always exists   a unimodular matrix $U(\bx,\delta)=\left[\begin{smallmatrix}
  U_1(\bx,\delta)\\
U_2(\bx,\delta)
 \end{smallmatrix}
 \right]\in \mathcal K^{n\times n}(\delta]$ such that  $ \smatrix{U_1(\bx,\delta)\\U_2(\bx,\delta)}L^{\dagger}(\bx,\delta)=\left[\begin{smallmatrix}
R(\bx,\delta)\\
 0
 \end{smallmatrix}
 \right]$ with $R(\bx,\delta)\in \mathcal K^{p\times p}(\delta]$ being of full rank over $\mathcal K(\delta]$. Then by $L(\bx,\delta)L^{\dagger}(\bx,\delta)={\rm I}_p$ and $U_1(\bx,\delta)L^{\dagger}(\bx,\delta)=R(\bx,\delta)$, we get that $U_1(\bx,\delta)=R(\bx,\delta)L(\bx,\delta)+T(\bx,\delta)U_2(\bx,\delta)$ for some matrix $T(x,\delta)\in \mathcal K^{p\times (n-p)}(\delta]$. It follows that $\left[\begin{smallmatrix}
  U_1(\bx,\delta)\\
U_2(\bx,\delta)
 \end{smallmatrix}
 \right]=\left[\begin{smallmatrix} R(\bx,\delta)L(\bx,\delta)+T(\bx,\delta)U_2(\bx,\delta)\\
 U_2(\bx,\delta) 
 \end{smallmatrix}
\right]=\left[\begin{smallmatrix}
{\rm I}&T(\bx,\delta)\\
 0&{\rm I} 
 \end{smallmatrix}
 \right]\left[\begin{smallmatrix}
 R(\bx,\delta)L(\bx,\delta) \\
 U_2(\bx,\delta) 
 \end{smallmatrix}
 \right]$ is unimodular and thus $\left[\begin{smallmatrix}
 R(\bx,\delta)L(\bx,\delta) \\
 U_2(\bx,\delta)
 \end{smallmatrix}
 \right]$ is unimodular as well. So $\kspan{R(\bx,\delta)L(\bx,\delta)\rd x}$ satisfies \textbf{(C)} by Theorem 13 of \cite{marquez2007new}. Notice that $\kspan{R(\bx,\delta)L(\bx,\delta)\rd x}$ and $\mathcal L$ have the same right-annihilator and  $\kspan{R(\bx,\delta)L(\bx,\delta)\rd x}\subseteq \mathcal L$. Hence   $\mathcal L$ is closed and its right-annihilator is causal. 
 
 (ii)$\Rightarrow$(iv): Assume that item (ii) holds. Define $\tilde x_1:=\lambda(\bx)+\eta(\boldsymbol{\theta}(\bx))$, where $\eta$ is any function in $\mathcal K^{p}$ of $\theta=[\theta_1,\ldots,\theta_{n-p}]^T$ and their delays, and $\tilde x_2:=\theta(\bx)$. Then  $\smatrix{\rd x_1\\ \rd x_2}=\smatrix{{\rm I}_p&E(\bx,\delta)\\0&{\rm I}_{n-p}}\Theta(\bx,\delta)\rd x$, where $ E(\bx,\delta)\rd \theta=E(\boldsymbol{\theta},\delta)\rd \theta=\rd\eta(\theta)$ and $\Theta(\bx,\delta)\rd x=\smatrix{\rd \lambda(\bx)\\ \rd \theta(\bx)}$. Since $\Theta(\bx,\delta)$ is unimodular as $\smatrix{\lambda(\bx)\\ \theta(\bx)}$ defines a bicausal change of coordinates, we have that $\smatrix{{\rm I}_p&E(\bx,\delta)\\0&{\rm I}_{n-p}}\Theta(\bx,\delta)$ is unimodular and $\smatrix{\tilde x_1\\\tilde x_2}$ defines a bicausal change of coordinates as well. Hence  by $\lambda(\tilde \bx_1,\tilde \bx_2)=\tilde x_1-\eta(\tilde\bx_2)$,  the matrix $L_1(\tilde\bx_1,\tilde\bx_2,\delta)={\rm I}_p$ is unimodular.
 
 (iv)$\Rightarrow$(iii): Suppose that item (iv) holds, then $\rd \lambda=L(\bx,\delta)\rd x=L(\bx,\delta)\Psi^{-1}(\bx,\delta)\Psi(\bx,\delta)\rd x=L(\bx,\delta)\Psi^{-1}(\bx,\delta)\smatrix{\rd \tilde x_1\\ \rd \tilde x_2}=\smatrix{L_1(\bx,\delta)&L_2(\bx,\delta)}\smatrix{\rd \tilde x_1\\ \rd \tilde x_2}$, where $\Psi(\bx,\delta)\rd x=\rd \varphi(\bx)$ and $\Psi(\bx,\delta)\in \mathcal K^{n\times n}(\delta]$ is unimodular. Because $L_1(\bx,\delta)$ is unimodular, we have
 $
 \smatrix{L_1(\bx,\delta)&L_2(\bx,\delta)}\smatrix{L^{-1}_1(\bx,\delta)\\0}=L(\bx,\delta)\Psi^{-1}(\bx,\delta)\smatrix{L^{-1}_1(\bx,\delta)\\0}={\rm I}_p.
 $
 It follows that $L^{\dagger}(\bx,\delta)=\Psi^{-1}(\bx,\delta)\smatrix{L^{-1}_1(\bx,\delta)\\0}$
  is a polynomial right-inverse of $L(\bx,\delta)$.
 
\end{proof}
The results of Theorem \ref{Thm:impl} can be  extended  to functions with dependent differentials via the results of (strong) integrability of left-submodules in \cite{kaldmae2015integrability}.  In the delay-free case \cite{conte2007algebraic}, for $s$-functions $ \lambda_k(x)\in \mathcal K$, $1\le k\le s$, if the rank of $\rd  \lambda$ over $\mathcal K$ is $p\le s$, then we can choose $p$-functions (whose differentials are independent over $\mathcal K$) from $ \lambda_k(x)$ as  new coordinates. While in the time-delay case, for functions with dependent differentials over $\mathcal K(\delta]$, even the conditions of Theorem \ref{Thm:impl} are satisfied, we can \emph{not} always choose $p$ functions from $\lambda_k(x,\delta)$ as new bicausal coordinates. For example, take $ \lambda_1(\bx_{1,[1]},\bx_{2,[2]})=   x_1(-1) +   x_2(-2)$ and $\lambda_2(\bx_{1,[1]},\bx_{2,[2]})=(x_1+ x_2(-1))(x_1(-1)+  x_2(-2))$, we have $\rd   \lambda_1=\delta \rd x_1+\delta^2\rd x_2$ and $\rd  \lambda_2=( x_1(-1) +x_2(-2) +(x_1+ x_2(-1)) \delta)\rd x_1+( ( x_1(-1)+ x_2(-2))\delta+   (x_1+  x_2(-1))\delta^2)\rd x_2$, it can be seen that $\rd \lambda_1$ and $\rd  \lambda_2$  are dependent over $\mathcal K(\delta]$, and the submodule $ {\mathcal L}=\kspan{\rd  \lambda_1,\rd \lambda_2}$ is closed and its right-annihilator $\kspan{\smatrix{\delta\\ -1}}$ is causal, but we can  \emph{not} choose  either $\lambda_1$  or  $\lambda_2$ as a new bicausal coordinate since neither $\kspan{\rd \lambda_1}$ nor $\kspan{\rd \lambda_2}$ is closed. Observe that we may still construct $\tilde \lambda=x_1+  x_2(-1)$ as a new bicausal  coordinate and $\tilde{\mathcal L} =\kspan{\rd \tilde \lambda}= {\mathcal L}$. In general,  the following results hold:
\begin{pro}\label{Pro:1}
Consider $s$-functions ${\lambda}_i(\bx)\in \mathcal K$, $1\le i\le s$, with $\dim {\mathcal L}=p\le s$, where $ {\mathcal L}:=\kspan{  \rd \lambda_k, 1\le k \le s}$.
If $ {\mathcal L}$ satisfies \textbf{(C)}, then we can find $p$-functions $\tilde \lambda_1,\ldots, \tilde\lambda_p\in \mathcal K$, which do not necessarily belong  to $\{ \lambda_1,\ldots, \lambda_s\}$, such that $  \tilde{\mathcal L} =\kspan{  \rd \tilde \lambda_k(\bx), 1\le k \le p}={\mathcal L}$ and $ \lambda(\bx)=[\lambda_1(\bx),\ldots,\lambda_s(\bx)]^T=0$ is equivalent to $\tilde \lambda(\bx)=[\tilde \lambda_1(\bx),\ldots,\tilde \lambda_p(\bx)]^T=0$, i.e., $x(t)$ satisfies $\lambda(\bx)=0$ if and only if it satisfies $\tilde\lambda(\bx)=0$.
\end{pro}
\begin{proof}
Choose any $p$-functions $ \lambda_1(\bx),\ldots, \lambda_p(\bx)$ from $\lambda(\bx)$ such that the differentials $\rd \lambda_k$, $1\le k\le p$,  are independent over $\mathcal K(\delta]$.
Then the submodule $\kspan{\rd \lambda_1,\ldots,\rd\lambda_p}$ is (strongly) integrable in the sense of \cite{kaldmae2015integrability}. Thus its closure $\overline{\kspan{\rd\lambda_1,\ldots,\rd \lambda_p}}$, which coincides with ${\mathcal L}$ (because ${\mathcal L}$ is closed), is (strongly) integrable as well by Lemma 2 of \cite{kaldmae2015integrability}. So there exist  $p$-functions $ \bar\lambda_1,\ldots, \bar\lambda_p$ such that $\kspan{  \rd  \bar\lambda_1, \ldots,  \rd \bar \lambda_p}=\overline{\kspan{\rd \lambda_1,\ldots,\rd\lambda_p}}={\mathcal L}$. However, it is not necessarily true that $\lambda(\bx)=0$ if and only if $\bar \lambda(\bx)=0$. Since ${\mathcal L}$ satisfies \textbf{(C)},  we can choose $x_1=\bar\lambda_1$, $\ldots$, $x_p=\bar\lambda_p$, $x_{p+1}= \theta_1$, $\dots$, $x_n=\theta_{n-p}$ as new bicausal coordinates by Theorem ~\ref{Thm:impl}. It follows that $\lambda_k$, $1\le k\le s$ depends only on $(x_1,\dots,x_p)$ and their delays, i.e., $\lambda=\lambda(\bx_1,\dots,\bx_p)$. For $1\le k\le p$, fix $x_k=x_k(-1)=\cdots=x_k(-\bar j)=c_k$, where $c_k$ are  constants, and solve the algebraic equations  $\lambda(c_1,\dots,c_p)=0$. 
Then by setting $\tilde \lambda_k=x_k-c_k=\bar \lambda_k(\bx)-c_k$, $1\le k\le p$, we have $\lambda(\bx_1,\dots,\bx_p)=0$ if and only if $\tilde \lambda(\bx_1,\dots,\bx_p)=0$.
\end{proof} 
We are now ready to present a generalization of the implicit function theorem for time-delay algebraic equations.
\begin{cor}[implicit function theorem]\label{Cor:2}
Consider $s$-algebraic equations $\lambda(\bx)=0$ and ${\mathcal L}:=\kspan{  \rd \lambda_k, 1\le k \le s}$. Let $\dim \mathcal L=p\le s$, if ${\mathcal L}$ satisfies \textbf{(C)}, then there exists a  bicausal change of coordinates $\smatrix{\tilde x_1\\
\tilde x_2}
 =\varphi(\bx)$  with $\tilde x_1\in \mathbb R^{p}$ and $\tilde x_2\in \mathbb R^{n-p}$ such that $\lambda(\tilde \bx_1,\tilde \bx_2)=0$  implies $\tilde x_1=\eta(\tilde\bx_2)$. 
\end{cor}
\begin{proof}
If $p<s$, then we use the results of Proposition \ref{Pro:1} to replace $ \lambda(\bx)=0$ by $\tilde \lambda(\bx)=0$. Because $\tilde {\mathcal L}=\mathcal L$ satisfies item (i) of Theorem~\ref{Thm:impl}, we have $\rd \tilde x_1=L^{-1}_1L_2(\tilde\bx_1,\tilde\bx_2,\delta)\rd \tilde x_2$ by item (iv). Hence by Poincar\'e lemma, there always exist functions  $\eta\in \mathcal K^p$ such that $ \tilde x_1=\eta(\tilde \bx_2)$.
\end{proof}
\begin{rem}\label{Rem:2}
The result of Corollary \ref{Pro:1} is sufficient but not necessary, take the following example, $\lambda(\bx_{1,[1]},\bx_{2,[1]})=(x_1+x_1(-1))/x_2(-1)+e=0$ with a constant $e\neq 0$, we have $$\rd \lambda=(\frac{1}{x_2(-1)}+\frac{1}{x_2(-1)}\delta)\rd x_1-(\frac{x_1+x_1(-1)}{x^2_2(-1)}\delta)\rd x_2.$$ 
It can be  seen  by using Lemma \ref{Lem:1} below that the right-annihilator of $\mathcal L=\kspan{\rd \lambda}$ is \emph{not} causal ($\Delta x_1=x_1(+1)$ is not causal). However,  $\lambda(\bx)=(x_1+x_1(-1))/x_2(-1)+e=0$ is equivalent to $\hat\lambda(\bx)=(x_1+x_1(-1))+ex_2(-1)=0$ (for $x_2(-1)\neq 0$), and $\kspan{\rd \hat\lambda}$ satisfies \textbf{(C)}. In fact, by the bicausal change of coordinates $\smatrix{\tilde x_1\\\tilde x_2}=\smatrix{x_1\\ x_1+ex_2}$, we have $\hat\lambda(\tilde \bx_1, \tilde \bx_2)=\tilde x_1+\tilde x_2(-1)=0$ implying that $\tilde x_1=-\tilde x_2(-1)$. Observe that  $\mathcal L$ does not satisfy \textbf{(C)} for all $x(t)$ but it satisfies \textbf{(C)} for all $x(t)$ such that $\lambda(\bx)=0$ because $\mathcal L$ restricted to $\{\bx\,|\, \lambda(\bx)=0\}$ is $\kspan{\frac{1+\delta}{x_2(-1)}\rd x_1+\frac{e}{x_2(-1)}\delta\rd x_2}$, which coincides with $\kspan{\rd \hat \lambda}$ and satisfies  \textbf{(C)}. Remark that when and how 
we can find $\hat\lambda$ in the general case is an interesting problem, but we will not discuss that in details as the purpose of  the remaining note is to show how to check the condition  of Corollary \ref{Pro:1} (section \ref{sec:4}) and to use it to solve DDAEs (section \ref{sec:5})
\end{rem}

\section{An  algorithm for checking  the condition  of the implicit function theorem}\label{sec:4}
To construct the right-annihilator of a left-submodule is, in general, not an easy task (see e.g., Remark \ref{Rem:1}), which makes the conditions of Theorem \ref{Thm:impl} and Corollary \ref{Cor:2}  difficult to be checked.  
 A conventional way to find  the kernel of a polynomial matrix-valued function $L(\bx,\delta)\in \mathcal K^{p\times n}(\delta]$ is to transform $L(\bx,\delta)$ into its Smith canonical form  $Q(\bx,\delta)L(\bx,\delta)P(\bx,\delta)=[L_1(\bx,\delta),0]$ by two (causal) unimodular matrices $Q$ and $P$ (see e.g., \cite{marquez2007new,califano2014coordinates,califano2016accessibility}). However, the existence of (causal) unimodular matrices to transform $L(\bx,\delta)$ into its Smith canonical form  requires already its kernel to be causal \cite{marquez2007new}. Therefore, the necessity of item (i) of Theorem \ref{Thm:impl} is uncheckable by the last method, i.e., if the kernel of $L(\bx.\delta)$ is not causal, we can not transform $L(\bx.\delta)$ into its Smith form via (causal) unimodular matrices in order to verify if the kernel is indeed not causal.

The following lemma provides some easily checkable necessary conditions  for the causality of the right-annihilator of
 a submodule generated by the differential of a function.
Consider a function $\lambda(\bz_{[0,\bar j]})\in \mathcal K$ of the variables $z=[z_1,\ldots,z_q]^T\in \mathbb R^q$ and its time-delays. Let   $\alpha\rd z= [\alpha_1,\ldots,\alpha_q]\rd z=\rd \lambda$, where $\alpha_i(\bz,\delta)=\sum\limits^{\bar j_i}_{j=0}{\alpha^{j}_i(\bz)}\delta^{j}\in \mathcal K(\delta]$, $1\le i \le q$, and denote $\bar j=\deg(\alpha)=\max\{ \bar j_i, 1\le i\le q\}$.
\begin{lem}\label{Lem:1}
If the right-annihilator of $\alpha(\bz,\delta)$  is causal, then 
there exists a permutation  of $\alpha_i$ (by that of $z_i$) such that $\alpha_1\not\equiv 0$, $0\le \bar j_1\le \bar j_2$ and the right-annihilator of $[\alpha_1(\bz,\delta),\alpha_2(\bz,\delta)]$ is causal as well.
Moreover, if that is causal, then   rewrite
$$[\alpha_1,\alpha_2]=[\bar \alpha_1,\bar \alpha_2]+\alpha^{\bar j_1}_1[\hat \alpha_1,\hat \alpha_2],$$
where $\hat \alpha_1(\delta)=\delta^{\bar j_1}$, $\hat \alpha_2(\bz,\delta)=\frac{\alpha^{\bar j_2}_2(\bz)}{\alpha^{\bar j_1}_1(\bz)}\delta^{\bar j_2}$,  we have that 
\begin{itemize}
\item [(i)] the delays of the variables $\bz$ of   $\hat\alpha_2(\bz,\delta)=\hat\alpha_2(\bz_{[\bar j_1,\bar j]},\delta)$ are at least $\bar j_1$, i.e., $ [\Delta^{\bar j_1}\hat \alpha_1, \Delta^{\bar j_1}\hat \alpha_2]=[1,\Delta^{\bar j_1}\hat \alpha_2]$ is causal.
    \item[(ii)] Let $(\xi_1,\xi_2)=\bz_{[\bar j_1,\bar j]}$ with  $\xi_2=(z_1(-\bar j_1), z_2(-\bar j_{2}))$. Then  by fixing $\xi_1$ as constants,  the codistribution $$\mathcal D:=\sspan{\rd z_1(-\bar j_1)+ \frac{\alpha^{\bar j_2}_2(\bz_{[\bar j_1,\bar j]})}{\alpha^{\bar j_1}_1(\bz_{[\bar j_1,\bar j]})}\rd z_2(-\bar j_2)}$$ is integrable. That is,  
    there exists a function $\hat \lambda(\bz_{[\bar j_1,\bar j]})\in \mathcal K$
    such that   
    \begin{align}\label{Eq:dist}
    \mathcal D=\sspan{\frac{\partial \hat \lambda(\bz_{[\bar j_1,\bar j]})}{\partial \xi_2}\rd \xi_2}.
    \end{align}

    \item[(iii)]  $\tilde z=\varphi(\bz)=[\tilde z_1,\ldots,\tilde z_q]^T$, where $$\tilde z_1=\Delta^{\bar j_1}\hat \lambda(\bz_{[\bar j_1,\bar j]}), \ \ \tilde z_2=z_2, \ \ \ldots, \ \ \tilde z_q=z_q,$$ defines a bicausal change of coordinates and $\alpha\rd z$ under $\tilde z$-coordinates, i.e., $\tilde \alpha\rd \tilde z=[\tilde \alpha_1,\ldots,\tilde \alpha_q]\rd \tilde z$  with $\tilde \alpha(\tilde{\bz},\delta)=\alpha(\bz,\delta)\Psi^{-1}(\bz,\delta)$, where $\Psi(\bz,\delta)\rd z=\rd \varphi(\bz)$, satisfies $\deg (\tilde \alpha_1)=\bar j_1$, $\deg (\tilde \alpha_2)<\bar j_2$ and $\deg (\tilde \alpha_i)=\bar j_i$ for $3\le i\le q$, that is, the polynomial degree of $\alpha_2$ is reduced by the bicausal coordinates change.
\end{itemize}

\end{lem} 
\begin{proof}
The proof of Lemma \ref{Lem:1} is given after that of Theorem \ref{Thm:alg}.
\end{proof} 
The above lemma shows a way to reduce the polynomial degree of the differential of a delayed function via  bicausal changes of coordinates. With the help of Lemma \ref{Lem:1} and inspired by the classical method to  transform  a polynomial matrix into its triangular normal form (or Hermite form, see e.g., \cite{gantmacher1959matrix}), we propose Algorithm \ref{Alg:1} below, which can be used to check the equivalent conditions of Theorem~\ref{Thm:impl} and to construct the desired complementary bicausal coordinates $(\theta_1,\ldots,\theta_{n-p})$.
\begin{algorithm} 
\caption{}\label{Alg:1}
\begin{algorithmic}[1]
\REQUIRE $\lambda_1(\bx),\ldots,\lambda_{p}(\bx)$
\ENSURE \emph{YES/NO}
\STATE Set $k \leftarrow 1$, $l\leftarrow 1$,  $q \leftarrow n$,   $z=[z_1,\ldots,z_q]^T\leftarrow [x_1,\ldots,x_n]^T$.
\IF{$k>1$}
\STATE  Fix $x_1,\ldots, x_{k-1}$ as constants, set $q\leftarrow n-k+1$ and set $z=[z_1,\ldots,z_q]^T\leftarrow [x_{k},\ldots,x_{n}]^T$ to regard 
$\lambda_k(\bx)=\lambda_k(\bx_1,\ldots, \bx_{k-1},\bz)=\lambda_k(\bz)$ as a function of $z$-variables and its time-delays.
\ENDIF
\STATE  Set  $\alpha(\bz,\delta)\rd z=\rd \lambda_k(\bz,\delta)$ to get   $\alpha=[\alpha_1,\ldots,\alpha_q]\in \mathcal K^q(\delta]$.
\STATE Find  a permutation matrix $P^l_k\in \mathbb R^{q\times q}$ such that  $\alpha_1\not\equiv 0$, $\bar j_1\le \bar j_2$ and $[\Delta^{\bar j_1}\hat\alpha_1, \Delta^{\bar j_1}\hat\alpha_2]$  is causal after $z\leftarrow P^l_kz$ and $\alpha\leftarrow \alpha(P^l_k)^{-1}$.
\IF{$\not\exists P^l_k$}
\RETURN \emph{NO}  
\ELSE 
\STATE Find $\hat \lambda(\bz_{[\bar j_1,\bar j]})\in \mathcal K$ such that (\ref{Eq:dist}) holds. 
\STATE Set $\tilde z_1\leftarrow\Delta^{\bar j_1}\hat \lambda$, $\tilde z_2\leftarrow z_2$, $\ldots$, $\tilde z_q\leftarrow z_q$.
\STATE Define a bicausal change  of $z$-coordinates $\tilde z=\varphi^l_k(\bz)=[\tilde z_1, \tilde z_2,\dots, \tilde z_q]^T\in\mathcal K^q$. 

\STATE Set $ \Psi_{k}(\bz,\delta)\rd z\leftarrow \rd \varphi^l_k(\bz)$, $\tilde \alpha(\bz,\delta)\leftarrow\alpha(\bz,\delta)\Psi_k^{-1}(\bz,\delta)$ and $z\leftarrow (\varphi^l_{k})^{-1}(\tilde \bz)$ to have 
  $\tilde \alpha(\tilde \bz,\delta)=[\tilde \alpha_1(\tilde \bz,\delta),\ldots,\tilde \alpha_q(\tilde \bz,\delta)]$ and $\lambda_k(\tilde \bz)$. 

\IF{$\exists 2\le i\le q: \tilde \alpha_i\not\equiv 0$} 
\STATE   Set  $\alpha\leftarrow \tilde \alpha$ and $z\leftarrow\tilde z$, $l\leftarrow l+1$ and go to line 5.
\ELSE 
\IF{$\deg(\tilde \alpha_1(\tilde \bz,\delta))\neq 0$} 
\RETURN \emph{NO}
\ELSE 
\IF{$k=p$}
\RETURN \emph{YES} 
\ELSE
\STATE Set  $[x_k,\ldots,x_n]^T\leftarrow [\tilde z_1,\dots,\tilde z_q]^T$.
\STATE $k\leftarrow k+1$, $l\leftarrow 1$.
 \STATE Go to line 2
 \ENDIF
  \ENDIF
 \ENDIF
\ENDIF
\end{algorithmic}
\end{algorithm} 
 \begin{thm}\label{Thm:alg}
The functions $\lambda_k(\bx)$,$1\le k\le p$, satisfy the equivalent conditions in Theorem \ref{Thm:impl} if and only if Algorithm \ref{Alg:1} returns to YES. Moreover, if Algorithm \ref{Alg:1} returns to YES, then  let $\tilde z_2,\ldots,\tilde z_q$ with $q=n-p+1$, be the functions from the last iteration, i.e., $[\tilde z_2,\ldots,\tilde z_q]^T=Q_p\varphi_p\circ\cdots\circ Q_1\varphi_{1}$, where, for each $1\le k\le p$,
\begin{align}\label{Eq:phik}
 \varphi_k=\varphi^{l_k}_k\circ P^{l_k}_k\cdots\varphi^2_k\circ P^2_k\varphi^1_k\circ P^1_k\in \mathcal K^{n-k+1} \end{align} 
and $Q_k=[0,{\rm I}_{n-k}]\in \mathbb R^{(n-k)\times (n-k+1)}$ selects the last $n-k$ rows of $\varphi_k$, and $l_k$ denotes the number of iterations for $\lambda_k$, we have that $[\lambda_1,\ldots, \lambda_{p}, \theta_1,\ldots, \theta_{n-p}]^T$ is a bicausal change of $x$-coordinates, where $\theta_1=\tilde z_2$, $\ldots$, $\theta_{n-p}=\tilde z_q$.
\end{thm}
\begin{rem}
Algorithm \ref{Alg:1} and Theorem \ref{Thm:alg}  provide  
 another way to prove $(i)\Rightarrow (ii)$ of Theorem \ref{Thm:impl}, the original proof in \cite{califano2016accessibility} uses a contradiction   with the help of the extended Lie brackets. Algorithm~\ref{Alg:1} proves $(i)\Rightarrow (ii)$ by  directly constructing the complementary bicausal coordinates $(\theta_1,\ldots,\theta_{n-p})$ in Theorem \ref{Thm:impl} (ii)  using  condition \textbf{(C)} and Lemma 
 \ref{Lem:1}.
\end{rem}
 \begin{proof}[Proof of Theorem \ref{Thm:alg}]
``\emph{Only if}:'' Assume that $\mathcal L$ satisfies \textbf{(C)}. Then by Theorem 13 of \cite{marquez2007new}, the latter assumption is equivalent to that there exists a matrix $\Theta(\bx,\delta)\in \mathcal K^{(n-p)\times n}(\delta]$  such that $\smatrix{L(\bx,\delta)\\ \Theta(\bx,\delta) }$ is unimodular, where $L(\bx,\delta) \rd x=\rd \lambda(\bx)$. It follows that $\mathcal L_k:=\kspan{\rd \lambda_1,\ldots, \rd \lambda _k}$ for all $1\le k\le p$ satisfy \textbf{(C)}  because we can always find $\Theta_k$ such that $\smatrix{L_k(\bx,\delta)\\ \Theta_k(\bx,\delta) }$ is unimodular, where $L_k(\bx,\delta)\rd x= \smatrix{\rd \lambda_1(\bx)\\ \dots\\ \rd \lambda_k(\bx)}$. Remark that the  property that $\mathcal L_k$ satisfies \textbf{(C)} is invariant under bicausal changes of  coordinates. Now consider $k=1$, i.e., in each $1\le l \le l_1$-iteration  of Algorithm \ref{Alg:1}, the right-annihilator of $\mathcal L_1$ is causal and thus by Lemma \ref{Lem:1},  we can always find $P^l_1$ such that $\Delta^{\bar j_1}[\hat \alpha_1,\hat \alpha_2]$ is causal. By reducing the polynomial degree of $\alpha_2$ and permutations,  $\tilde \alpha_i$ for all $i\ge 2$ eventually become  $0$ at $l=l_1$. Moreover, $\deg(\tilde \alpha_1)=0$ for $l=l_1$
as $\mathcal L_1$ is closed. So the algorithm does not returns to \emph{NO} in the the first $l_1$-iterations. Suppose that the algorithm does not return to \emph{NO} for $k=1,\ldots,k^*-1$, i.e., after $(l_1+\cdots+l_{k^*-1})$-iterations, then $\smatrix{\lambda_1\\ \cdots\\ \lambda_{k^*}}$~becomes $\smatrix{\lambda_1(\bx_1)\\  \cdots \\ \lambda_{k^*-1}(\bx_1,\ldots,\bx_{k^*-1})\\  \lambda_{k^*}(\bx_1,\ldots,\bx_{k^*-1},\bz_{1},\ldots,\bz_{q}) }$ with $q=n-k^*+1$ and  
$
\smatrix{\rd\lambda_1 \\  \cdots \\ \rd \lambda_{k^*-1}\\  \rd \lambda_{k^*}}=\smatrix{c_1&0&\cdots&0&\cdots&0 \\ \vdots&\ddots&\vdots&\vdots&\cdots&\vdots \\
\star&\cdots&c_{k^*-1}&0&\cdots&0\\
\star&\cdots&\star&\alpha_1&\cdots&\alpha_q
}\smatrix{\rd x_1\\  \vdots \\ \rd x_{k^*-1}\\ \rd z_1\\ \vdots \\ \rd z_q},
$
where $c_k\not\equiv 0$ and $\deg(c_k)=0$ for all $1\le k\le k^*-1$, ``$\star$'' denotes some irrelevant terms.
Thus by $\mathcal L_{k^*}$ satisfies \textbf{(C)},   we get that  $\kspan{\alpha(\bz)\rd z}$ satisfies \textbf{(C)} (when fixing $(x_1,\ldots,x_{k^*-1})$ as constants) as well, which indicates that Algorithm \ref{Alg:1} does not return to \emph{NO} for $k=k^*$. Hence the algorithm returns to \emph{YES} once $k=p$.

``\emph{If}:'' Suppose that the algorithm returns to \emph{YES}. Then, we can construct the following bicausal changes of  $x$-coordinates:
$$
\psi_1=\varphi_1, \  \psi_2=\smatrix{M_1\psi_1\\ \varphi_2(N_1\psi_1)}, \ \ldots, \ \psi_p=\smatrix{M_{p-1}\psi_{p-1}\\ \varphi_p(N_{p-1}\psi_{p-1})}, 
$$
where $\varphi_k$, $1\le k\le p$, are defined by (\ref{Eq:phik}), $M_k=\smatrix{{\rm I}_k&0}\in \mathbb R^{k\times n}$ and $N_k=\smatrix{0&{\rm I}_{n-k}}\in \mathbb R^{(n-k)\times n}$. Indeed, each $\psi_k$ defines a bicausal change of coordinates on $\mathcal K^n$ because 
$\rd \psi_k=\smatrix{{\rm I}_k&0\\ \star &\Psi_k}\rd \psi_{k-1}$, where $\Psi_k(\bx_1,\ldots,\bx_{k-1},\bz,\delta)\rd z=\rd \varphi_k(\bx_1,\ldots,\bx_{k-1},\bz)$ and $\Psi_k$, by the constructions in Algorithm~\ref{Alg:1}, is unimodular. Then define the following bicausal change of coordinates $\tilde x=[\tilde x_1,\ldots,  \tilde x_n]^T=\psi(\bx)=\psi_p\circ\cdots\circ\psi_1(\bx)$, thus in $\tilde x$-coordinates  we have  that
$$
\smatrix{\rd\lambda_1 \\  \cdots  \\  \rd \lambda_{p}}=\smatrix{c_1&0&\cdots&0&\cdots&0 \\ \vdots&\ddots&\vdots&\vdots&\cdots&\vdots \\
*&\cdots&c_{p}&0&\cdots&0\\ 
}\smatrix{\tilde x_1\\ \vdots \\ \tilde x_p\\ \tilde x_{p+1} \\ \vdots \\ \tilde x_n},
$$
where $c_i=c_i(\tilde \bx)\not\equiv 0$ and $\deg(c_i)=0$ for  all $1\le i \le p$. It follows that $[\lambda_1,\ldots,\lambda_p,\tilde x_{p+1},\ldots,\tilde x_n]^T$ is a bicausal change of coordinates because $T(x,\delta)$, where $T\rd \tilde x=  \smatrix{\rd \lambda\\ \rd\tilde x_{p+1}\\ \vdots \\ \rd \tilde x_n}$, is   a unimodular matrix. Thus item (ii) of Theorem \ref{Thm:impl} holds with $\theta_1=\tilde x_{p+1}$, $\ldots$, $\theta_{n-p}=\tilde x_n$. Moreover, by using $\varphi_k$ and $Q_k$,  we can express $[\tilde x_{p+1},\ldots,\tilde  x_n]^T=Q_p\varphi_p\circ\cdots\circ Q_1\varphi_{1}$.
 \end{proof}
\begin{proof}[Proof of Lemma \ref{Lem:1}]
We need to prove that there exist two integers $1\le r\le q-1$ and  $r+1\le s\le q$ such that the right-annihilator of $[\alpha_r,\alpha_s]$ is causal. Suppose that   the right-annihilators of $[\alpha_r,\alpha_s]$ are not causal for all $1\le r\le q-1$, $r+1\le s\le q$.  Let $\smatrix{\beta_{l(r,s)}\\ \gamma_{l(r,s)}}\in \mathcal K^2(\delta]$, where $l(r,s)=(r-1)(q-\frac{r}{2})+s-r$ and $1\le l\le l^*= \frac{q(q-1)}{2} $, be a basis for the right-annihilator  of $[\alpha_r,\alpha_s]$, then define $\tau_{l}:=[0,\dots,0,\beta_{l},0,\dots,0,\gamma_{l},0,\dots,0]^T$, where $\beta_{l}$ and $\gamma_{l}$ are in the $r$-th and $s$-th rows of $\tau_{l}$, respectively. It follows that $\alpha\tau_{l}=0$ for all $1\le l\le l^*$. Thus the right-submodule $\mathcal T=\kspan{\tau_1,\ldots,\tau_{l^*}}$ is in the right-annihilator of $\kspan{\alpha\rd z}$, so $\dim \mathcal T\le q-1$. Recall that the right-annihilator of a left-submodule is always closed (see \cite{califano2021nonlinear}). By the construction,  $\mathcal T$ is closed and $\dim \mathcal T\ge q-1$, which implies $\mathcal T$ coincides with  the right-annihilator of $\kspan{\alpha\rd z}$ because they have the same dimension $q-1$ and are both closed. If  $\tau_{l}$, for all $1\le l\le \frac{q(q-1)}{2}$, are not causal, we have that the right-annihilator of $\alpha$ is not causal. Hence if the right-annihilator of $\alpha$ is causal, then there must exist $r,s$ such that $\tau_{l}$ is causal.
  
  (i) If the right-annihilator of  $[\alpha_1(\bz,\delta),\alpha_2(\bz,\delta)]$, generated by $\smatrix{\beta(\bz,\delta)\\ \gamma(\bz,\delta)}$,  is causal, then  the right-annihilator of $[\alpha^{\bar j_1}_1(\bz)\delta^{\bar j_1},\alpha^{\bar j_2}_2(\bz)\delta^{\bar j_2}]$ is also causal. Indeed, write $\beta(\bz,\delta)=\sum\limits^{\bar j_{\beta}}_{j=1}\beta^j(\bz)\delta^{j}$ and $\gamma(\bz,\delta)=\sum\limits^{\bar j_{\gamma}}_{j=1}\gamma^j(\bz)\delta^{j}$, we can deduce both $\bar j_1+\bar j_{\beta}=\bar j_2+\bar j_{\gamma}$ and $[\alpha^{\bar j_1}_1(\bz)\delta^{\bar j_1},\alpha^{\bar j_2}_2(\bz)\delta^{\bar j_2}]\smatrix{\beta^{\bar j_{\beta}}(\bz)\delta^{\bar j_{\beta}}\\ \gamma^{\bar j_{\gamma}}(\bz)\delta^{\bar j_{\gamma}}}=0$ from $\alpha_1\beta+\alpha_2\gamma=0$. It follows that the right-annihilator of $[\hat \alpha_1,\hat\alpha_2]$ is causal. Then because $\hat \alpha_1=\delta^{\bar j_1}$, by a direct calculation, the right-annihilator of $[\hat \alpha_1,\hat\alpha_2]$ is generated by $[\Delta^{\bar j_1}\hat\alpha_2,-1]^T$. Hence  $[\Delta^{\bar j_1}\hat \alpha_1, \Delta^{\bar j_1}\hat \alpha_2]=[1,\Delta^{\bar j_1}\hat\alpha_2]$ is causal.

(ii)
Let $(\tilde \xi_1,\xi_2)=\bz_{[0, \bar j]}$ and $\xi_2=(z_1(-\bar j_1), z_2(-\bar j_2))$.
If we fix  $\tilde \xi_1$ as constants, then $\lambda(\bz_{[0, \bar j]})=\lambda(\tilde \xi_1,\xi_2)=\lambda(\xi_2)$ can be seen as a function of $\xi_2$.
It follows that the one form $\hat \omega=  \alpha^{\bar j_1}_1(\tilde \xi_1,\xi_2) \rd z_1(-\bar j_1)+\alpha^{\bar j_2}_2(\tilde \xi_1,\xi_2) \rd z_2(-\bar j_1)$ is exact (by fixing $\tilde \xi_1$). Then by Frobenius theorem, the codistribution 
$
\sspan{  \rd z(-\bar j_1)+\hat \alpha^{\bar j_2}_2(\bz_{[\bar j_1,\bar j]})\rd z(-\bar j_2)}=\sspan{\hat \omega}
$, where $\hat \alpha^{\bar j_2}_2=\frac{ \alpha^{\bar j_2}_2}{ \alpha^{\bar j_1}_1}$ depends only on $\bz_{[\bar j_1,\bar j]}$ by item (i), is integrable when fixing  $\xi_1$, where $(\xi_1,\xi_2)=\bz_{[\bar j_1,\bar j]}$. Hence there exists a function $\hat \lambda=\hat \lambda(\bar z_{[\bar j_1,\bar j]})\in \mathcal K$ such that (\ref{Eq:dist}) holds

(iii) By construction, we have  $\rd \hat \lambda(\bz_{[\bar j_1,\bar j]})=\hat\beta(\bz_{[\bar j_1,\bar j]},\delta)  \rd z+c\hat \alpha_1(\bz_{[\bar j_1,\bar j]},\delta)  \rd z_1+c\hat \alpha_2(\bz_{[\bar j_1,\bar j]},\delta)  \rd z_2$ for some function $c=c(\bz_{[\bar j_1,\bar j]})\in \mathcal K$, where $\hat\beta=[\hat \beta_1,\ldots,\hat \beta_q]$, $\hat \beta_1\equiv 0$, $\deg(\hat\beta_2)\le \bar j_2-1$,  and for $3\le i\le q$,
$\deg(\hat\beta_i)= \bar j_i$ if $\bar j_i\ge \bar j_1$ and $\hat\beta_i\equiv 0$ if $\bar j_i< \bar j_1$. Now let  $\varphi(\bz)=[\Delta^{\bar j_1}\hat \lambda(\bz), z_2,\ldots,z_q]^T$ and $\Psi\rd z=\rd \varphi(\bz)$, we get $$\Psi=\smatrix{\Delta^{\bar j_1}(c\hat \alpha_1)& \Delta^{\bar j_1}(\hat \beta_2+c\hat \alpha_2)&\Delta^{\bar j_1}\hat\beta_3& \cdots&\Delta^{\bar j_1}\hat \beta_q\\
0&1&0&\cdots&0\\
0&0&1&\cdots&0\\
0&0&0&{\rm I}&0\\
0&0&0&\cdots&1
},$$ which is upper triangular and $\Delta^{\bar j_1}(c\hat \alpha_1)=\Delta^{\bar j_1}(c\delta^{\bar j_1})=c(+{\bar j_1})$ is of polynomial degree zero, and thus $\Psi(\bz,\delta)$ is unimodular and $\varphi(\bz,\delta)$ is a bicausal change of coordinates. Then we have $[\tilde \alpha_1,\ldots,\tilde \alpha_q]=\alpha \Psi^{-1}=$
$$
\smatrix{\bar \alpha_1+\alpha^{\bar j_1}_1\delta^{\bar j_1}\\\bar\alpha_2+\alpha^{\bar j_1}_1\hat \alpha_2\\ \alpha_3\\\cdots\\   \alpha_q}^T\smatrix{\frac{1}{c(+\bar j_1)}& -(\frac{\Delta^{\bar j_1}\hat \beta_2}{c(+\bar j_1)}+\Delta^{\bar j_1}\hat \alpha_2)&-\frac{\Delta^{\bar j_1}\hat \beta_3}{c(+\bar j_1)}& \cdots&-\frac{\Delta^{\bar j_1}\hat \beta_q}{c(+\bar j_1)}\\
0&1&0&\cdots&0\\
0&0&1&\cdots&0\\
0&0&0&{\rm I}&0\\
0&0&0&\cdots&1
}
$$
 By a direct calculation, we have $\tilde\alpha_1=\frac{\bar \alpha_1}{c(+\bar j_1)}+\frac{\alpha^{\bar j_1}_1}{c}\delta^{\bar j_1}$,  $\tilde\alpha_2=-(\bar \alpha_1+\alpha^{\bar j_1}_1\delta^{\bar j_1})\frac{\Delta^{\bar j_1}\hat \beta_2}{c(+\bar j_1)}-\bar \alpha_1\Delta^{\bar j_1}\hat \alpha_2+\bar \alpha_2$ and $\tilde \alpha_i=\alpha_i-(\bar \alpha_1+\alpha^{\bar j_1}_1\delta^{\bar j_1})(\frac{\Delta^{\bar j_1}\hat \beta_i}{c(+\bar j_1)})$ for $3\le i\le q$. Notice that $\deg(\bar \alpha_1+\alpha^{\bar j_1}_1\delta^{\bar j_1})=\bar j_1$, $\deg(\frac{\Delta^{\bar j_1}\hat \beta_2}{c(+\bar j_1)})\le \bar j_2-1-\bar j_1$ and $\deg(-\bar\alpha_1\Delta^{\bar j_1}\hat \alpha_2+\bar \alpha_2)\le \bar j_2-1$. Hence  $\deg (\tilde \alpha_1)=\bar j_1$, $\deg (\tilde \alpha_2)<\bar j_2$ and $\deg (\tilde \alpha_i)=\deg(\alpha_i)=\bar j_i$,  $\forall \, i\ge 3$.

\end{proof}
\begin{example}\label{Ex:alg1}
1).  Consider $b(\bx)=0$ in section \ref{sec:1}, we apply  Algorithm~\ref{Alg:1} to $b(\bx)$.  For $k=1$, $l=1$,  $$
\alpha=\smatrix{x_2(-1)+x_2x_2(-2)\delta,&x_1(-1)x_2(-2)+x_1\delta+x_1(-1)x_2\delta^2}.
$$
It is seen that $P^1_1={\rm I}_2$ and   $[\hat\alpha_1,\hat \alpha_2]=[\delta,\frac{x_1(-1)}{x_2(-2)}\delta^2]$. Thus $\Delta^{\bar j_1}[\hat\alpha_1,\hat \alpha_2]=[1,\frac{x_1}{x_2(-1)}\delta]$ with $\bar j_1=1$ is causal. Then $\sspan{\rd x_1(-1)+\frac{x_1(-1)}{x_2(-2)}\rd x_2(-2)}$ is integrable and we   find the function $\hat \lambda=x_1(-1)x_2(-2)$ satisfying (\ref{Eq:dist})   ($\xi_1$ is absent and $\xi_2=(x_1(-1),x_2(-2))$). The bicausal coordinate transformation is $\varphi^1_1=\smatrix{\tilde z_1\\\tilde z_2}=\smatrix{x_1x_2(-1)\\  x_2}$ as $\Delta^1\hat \lambda=x_1x_2(-1)$. Thus in $\tilde z=(\tilde z_1,\tilde z_2)$-coordinates,  $b=b(\tilde \bz_1,\tilde \bz_2)=\tilde z_1+\tilde z_1(-1)\tilde z_2$ and $\tilde \alpha=[1+\tilde z_2\delta,\tilde z_1(-1)]$. So $\tilde \alpha_2\not\equiv 0$ and we go to the second iteration (i.e., line 15$\to$line 5).
For $k=1$, $l=2$, we use the permutation matrix $P^2_1=\smatrix{0&1\\
1&0}$  to have $\tilde \alpha\smatrix{\rd \tilde z_1\\ \rd \tilde z_2}=\smatrix{\tilde z_1(-1)&1+\tilde z_2\delta}\smatrix{\rd \tilde z_2\\ \rd \tilde z_1}$. Define new coordinates  $\smatrix{\tilde x_1\\ \tilde x_2}=P^2_1\smatrix{\tilde z_1\\ \tilde z_2}$ to have $b(\tilde x_1,\tilde x_2)=\tilde x_2+\tilde x_2(-1)\tilde x_1+e_2$ and $\rd b=\smatrix{\tilde x_2(-1)&1+\tilde x_1\delta}\smatrix{\rd \tilde x_1\\ \rd \tilde x_2}$. Now although $1+\tilde x_1\delta\not\equiv 0$, we can already conclude that $b(\tilde x_1,\tilde x_2)$ satisfies item (iv) of Theorem~\ref{Thm:impl} without continuing the algorithm because $\tilde x_2(-1)$ is unimodular. Moreover, we get $\varphi_1=P^2_1\varphi^1_1$ by (\ref{Eq:phik}) and the complementary coordinate $\theta=Q_1\varphi_1=x_1x_2(-1)$. It can be checked that $\smatrix{b(\bx)\\\theta(\bx)}$ is indeed a bicausal change of coordinates. Moreover, by Corollary~\ref{Cor:2}, $b(\bx)=0$ implies   $\tilde x_1=\frac{-e_2-\tilde x_2}{\tilde x_2(-1)}$.

2). Consider $c(\bx)=0$ in section \ref{sec:1} and apply Algorithm~\ref{Alg:1} to $c(\bx)$. For $l=1$, $\alpha=[x_1(-1)+x_1\delta,x_2(-1)+x_2\delta]$ and $\hat\alpha=[\delta,\frac{x_2}{x_1}\delta]$, it is seen that $\bar j_1=1$ and $\Delta^1\hat\alpha=[1,\frac{x_1(+1)}{x_2(+1)}]$ is not causal. Thus  Algorithm \ref{Alg:1} returns to \emph{NO}, meaning that $c(\bx)$ can not be regarded as a bicausal coordinate and there does not exist a bicausal coordinates transformation such that Theorem~\ref{Thm:impl} (iv) holds.

3). As the third example, we consider two functions together:
$$
\left\{ \begin{aligned}
\lambda_1&=x_2x_1(-2)+x_3(-1)x_2(-1)\\
\lambda_2&=x_3(-1)x_2(-1)x_1(-1)+x_2x_1(-2)x_1+x_3(-1)x_2(-1)x_1
\end{aligned}\right.
$$
and apply Algorithm \ref{Alg:1}. For $k=1$, $l=1$, $\alpha=[x_2\delta^2,x_1(-2)+x_3(-1)\delta,x_2(-1)\delta]$, we find $P^1_1=\smatrix{0&0&1\\
0&1&0\\1&0&0}$ and $\alpha (P^1_1)^{-1}P^1_1\rd x=\smatrix{x_2(-1)\delta&x_1(-2)+x_3(-1)\delta&x_2\delta^2}\smatrix{\rd x_3\\\rd x_2\\\rd x_1}$. Thus we have $\alpha_1=x_2(-1)\delta$, $\alpha_2=x_1(-2)+x_3(-1)\delta$ and $\bar j_1=\bar j_2=1$. So $[\Delta^1\hat \alpha_1,\Delta^1\hat \alpha_2]=[1,\frac{x_2}{x_3}]$ is causal. Then we find $\hat \lambda=x_2(-1)x_3(-1)$ and $\varphi^1_1=[\tilde x_1,\tilde x_2,\tilde x_3]^T= [x_2x_3,x_2,x_1]^T$ to get 
$\tilde \alpha=[\delta,\tilde x_3(-2),\tilde x_2\delta^2]$ and $\lambda_1=\tilde x_1(-1)+\tilde x_2\tilde x_3(-2)$. Since both $\tilde \alpha_2$ and $\tilde \alpha_3$ are not zero, we drop all the tildes and go to next iteration (i.e., line 15$\to$line 5).  For $k=1$, $l=2$, $\lambda_1=x_1(-1)+ x_2 x_3(-2)$, we find $P^2_1=\smatrix{0&1&0\\
1&0&0\\0&0&1}$ and $\varphi^2_1=[\tilde x_1,\tilde x_2,\tilde x_3]=[x_2x_3(-2)+x_1(-1), x_1,x_3]^T$. Then $\tilde \alpha=[1,0,0]$, we have $\deg(\tilde \alpha_1)=0$ and go to $k=2$ (i.e., line 25$\to$line 2). Notice that $\varphi_1=\varphi^2_1\circ P^2_1\varphi^1_1\circ P^1_1=[\tilde x_1,\tilde x_2,\tilde x_3]^T=[x_2x_1(-2)+x_2(-1)x_3(-1),x_2x_3,x_1]^T$ and in $\varphi_1$-coordinates, we have $\lambda_2=\tilde x_1\tilde x_3+\tilde x_2(-1)\tilde x_3(-1)$. 

Now we are at line 2 and we restart the procedure. For $k=2$, $l=1$, set $z_1=\tilde x_2$ and $z_2=\tilde x_3$ to have $\lambda_2(z_1,z_2)=\tilde x_1z_2+ z_1(-1)z_2(-1)$ and $\alpha=[z_2(-1)\delta,\tilde x_1+z_1(-1)\delta]$. We find  $P^1_2={\rm I}_2$ and $\varphi^1_2=\smatrix{\tilde z_1\\ \tilde z_2}=\smatrix{z_1z_2\\z_2}$. It follows that $\tilde \alpha=[\delta,\tilde x_1]$ and $\lambda_2=\tilde z_1(-1)+\tilde x_1\tilde z_2$. Drop the tildes of $\tilde z_1(-1)$ and $\tilde z_2(-1)$.  For $k=2$, $l=2$, $\lambda_2= z_1(-1)+\tilde x_1 z_2$, we find $P^2_2=\smatrix{0&1\\
1&0}$ and $\varphi^2_2= \smatrix{ z_1(-1)+\tilde x_1z_2\\z_1}$. Thus $\tilde \alpha=[1,0]$ and the algorithm returns to \emph{YES}. Moreover, we have $\varphi_2=\varphi^2_2\circ P^2_2\varphi^1_2\circ P^1_2=\smatrix{\tilde x_1z_2+ z_1(-1)z_2(-1)\\ z_1}$. Thus the complementary coordinate $\theta=z_1=\tilde x_2=x_2x_3$, we can check that $\smatrix{\rd \lambda_1 \\ \rd \lambda_2 \\ \rd \theta}$ is indeed a unimodular matrix.
\end{example}

\section{Applications to  nonlinear differential-algebraic equations with time-delays}\label{sec:5}
Consider a delayed differential-algebraic equation  (DDAE) of the following form:
\begin{align} \label{Eq:DDAE}
 	\Xi:  \quad  \begin{aligned}
  \sum\limits_{j=0}^{\bar j}E^j(\bx_{[\bar i]})\delta^j\dot x=F(\bx_{[\bar i]}) 
 	\end{aligned}
\end{align}
with an initial-value function  	$x(s)=\xi_x(s), \  s\in [-\bar i,0]$, where $E^j:\mathbb R^{(\bar i+1)n}\to \mathcal K^{p\times n}$ and $F:\mathbb R^{(\bar i+1)n}\to \mathcal K^{p}$, where $\bar i$ and $\bar j$ denote the maximal delay of $x$ and $\dot x$, respectively. We can shortly rewrite (\ref{Eq:DDAE}) as $E(\bx,\delta)\dot x=F(\bx)$, where $E(\bx,\delta)=\sum\limits_{j=0}^{\bar j}E^j(\bx_{[\bar i]})\delta^j\in \mathcal K^{p\times n}(\delta]$. 
Remark that  the form (\ref{Eq:DDAE}) is general enough to describe many constrained physical models under delay effects as in  e.g. \cite{venkatasubramanian1994time,ascher1995numerical,ha2016analysis}, the DDAE $\Xi$  reduces to a delay-free DAE of the form $E(x)\dot x=F(x)$   \cite{RabiRhei02,kunkel2006differential,chen2022geometric} when $\bar j=\bar i=0$.
\begin{defn} 
 A function $x:\mathbb R\to \mathbb R^n$ is a   solution of $\Xi$ with the initial-value function $\xi_x$  if there exists   $T>0$  such that $x(t)$ is continuously differentiable on $[-\bar i,T)$ and satisfies (\ref{Eq:DDAE}) for all $t\in [0,T)$. 
 \end{defn}
 We will call $\Xi$ an index-$0$ DDAE if $E(\bx,\delta)$ is of full row rank over $\mathcal K(\delta]$. An index-$0$ DDAE is very close to a delayed  ODE, the latter  has the classifications of  the retarded, the neutral and the advanced types, and can be solved via the step method (see e.g., \cite{fridman2014introduction}). For an index-$0$ DDAE $\Xi$, if $p=n$ and $\rk{E^0(\bx)}=p$, we can always rewrite $\Xi$ as a delayed ODE of the neutral type: $$\dot x=(E^0)^{-1}F(\bx)- \sum\limits_{j=1}^{\bar j}(E^0)^{-1}E^j(\bx)\delta^j\dot x=f(\bx_{[0,\bar i]},\dot \bx_{[1,\bar j]}).$$ 
 Remark that if $\rk_{\mathcal K}{E^0(\bx)}\neq p$, then an index-$0$ DDAE results in a mixed type, or in particular, an advanced type delayed ODE, for which, in general, is hard to define a smooth solution unless the 
initial-value functions satisfy some restrictive conditions. In the present note, we are only interested in  delayed ODEs of neutral or retarded types, so below we make the assumption that $\rk_{\mathcal K}{E^0(\bx)}= p$ for index-$0$ DDAEs.

Now given a DDAE $\Xi: E(\bx,\delta)\dot x=F(\bx)$, which may not be index-$0$, we propose the following algorithm  to reduce its index   with the help of the results in sections \ref{sec:3} and \ref{sec:4}. Algorithm \ref{Alg:2} generalizes the geometric reduction algorithm for delay-free DAEs in~\cite{chen2022geometric}. 
\begin{algorithm}
\caption{DDAE reduction algorithm}\label{Alg:2}
\begin{algorithmic}[1]
\REQUIRE $E(\bx,\delta)$ and $F(\bx)$
\ENSURE $E_{k^*}(\bz_{k^*},\delta)$ and $F_{k^*}(\bz_{k^*})$
\STATE Set $k \leftarrow 0$,    $z_k\leftarrow x$, $E_k \leftarrow E$, $F_k \leftarrow F$, $r_{k-1}=p$, $n_{k-1}=n$ 
\IF{$\rk_{\mathcal K(\delta]}E_k(\bz_k,\delta)=r_{k-1}$ }
\RETURN $k^*\leftarrow k$, $z_{k^*}\leftarrow z_k$, $E_{k^*}\leftarrow E_{k}$,  $F_{k^*}\leftarrow F_{k}$
\ELSE 
\STATE Denote $\rk_{\mathcal K(\delta]} E_{k}(\bz_k,\delta)=r_k<r_{k-1}$. 
\STATE Find a  unimodular matrix $Q_k(\bz_k,\delta)\in \mathcal K^{r_{k-1}\times r_{k-1}}(\delta]$   such that $Q_k(\bz_k,\delta)E_k(\bz_k,\delta)=\smatrix{E_{k1}(\bz_k,\delta)\\ 0}$, where $E_{k1}(\bz_k,\delta)\in \mathcal K^{r_{k}\times r_{k-1}}(\delta]$   and $\rk_{\mathcal K(\delta]}E_{k1}(\bz_k,\delta)=r_k$.
\STATE  Denote $Q_{k}(\bz_k,\delta)F_k(\bz_k)=\smatrix{F_{k1}(\bz_k)\\ F_{k2}(\bz_k)}$, where $F_{k2}(\bz_k)\in \mathcal K^{r_{k-1}-r_{k}}$.
\STATE Define the submodule $\mathcal F_k:=\kspan{\rd F_{k2}(\bz_k)}$. Denote $\dim \mathcal F_k=n_{k-1}-n_{k}\le r_{k-1}-r_{k}$.
\STATE Assume that $\mathcal F_k$ satisfies \textbf{(C)}.
\IF{$n_{k-1}-n_{k}<r_{k-1}-r_{k}$} 
\STATE   Find functions $\tilde F_{k2}(\bz_k)\in \mathcal K^{n_{k-1}-n_{k}}$ such that $\kspan{\rd \tilde F_{k2}}=\mathcal F_{k}$ and $\tilde F_{k2}(\bz_k)=0$ is equivalent to $F_{k2}(\bz_k)=0$.
\STATE $F_{k2}\leftarrow \tilde F_{k2}$
\ENDIF
\STATE Find functions $\theta(\bz_{k})\in \mathcal K^{n_k}$ such that $\smatrix{z_{k+1}\\ \bar z_{k+1}}=\varphi_k(\bz_k)=\smatrix{\theta(\bz_k)\\ F_{k2}(\bz_k)}$ is a bicausal change of $z_k$-coordinates.
\STATE Set $[\tilde E_{k1}, \tilde E_{k2}]\leftarrow E_{k1}\Psi^{-1}_k$ and $z_k\leftarrow \varphi^{-1}_k(\bz_{k+1},\bar \bz_{k+1})$, where $\Psi_k\rd z_k=\rd \varphi_k$ and $\tilde E_{k1}(\bz_{k+1},\bar \bz_{k+1},\delta)\in \mathcal K^{r_{k}\times n_k}(\delta]$.
\STATE Set $E_{k+1}(\bz_{k+1},\delta)\leftarrow \tilde E_{k1}(\bz_{k+1},0,\delta)$ and $F_{k+1}(\bz_{k+1})\leftarrow  F_{k1}(\bz_{k+1},0)$.
\STATE Set $k\leftarrow k+1$ and go to line 2.
\ENDIF
\end{algorithmic}
\end{algorithm} 
\begin{thm}\label{Thm:DDAE}
Consider a DDAE $\Xi$, given by (\ref{Eq:DDAE}). Assume that the submodule $\mathcal F_k$ of Algorithm \ref{Alg:2} satisfies \textbf{(C)} for $k\ge 0$. We have that 
\begin{itemize}
\item[(i)] there exists an integer $0\le k^*\le p$ such that Algorithm \ref{Alg:2} returns to $E_{k^*}(\bz_{k^*},\delta)$ and $F_{k^*}(\bz_{k^*})$.

\item[(ii)] The DDAE $$\Xi^*: \ E_{k^*}(z_{k^*},\delta)\dot z_{k^*}=F_{k^*}(\bz_{k^*}) $$  is index-0, and  $\Xi^*$ and $\Xi$ have isomorphic solutions, i.e., there exists a bicausal change of coordinates $\Phi(\bx)=[z_{k^*},\bar z_{k^*},\ldots,\bar z_1]^T$ such that $z_{k^*}(t)$ is a solution of $\Xi^*$ with the initial-value function $\xi_{z_{k^*}}$ if and only if $x(t)=\Phi^{-1}(\bz_{k^*}(t),0,\ldots,0)$ is a solution of $\Xi$ with the initial-value function $\xi_{x}=\Phi^{-1}(\boldsymbol{\xi}_{z_{k^*}},0,\ldots,0)$.

\item[(iii)] For $E_{k^*}(\bz_{k^*},\delta)=\sum\limits^{\bar j_{z^*}}_{j=0}E^j_{k^*}(\bz_{k^*})\delta^j$, suppose that $\rk_{\mathcal K}E^0_{k^*}(\bz_{k^*})=r_{k^*}$, then $\Xi$ has a unique solution with the initial-value function $\xi_x$ if and only if $r_{k^*}=n_{k^*}$.

\end{itemize}
\end{thm}
\begin{proof}
(i) By using the results of Lemma 4 of \cite{marquez2007new}, Corollary \ref{Pro:1} and Theorem \ref{Thm:impl} above, respectively, we can guarantee the existences of the unimodular matrix $Q_k(\bz_k,\delta)$ of line 6, the functions $\tilde F_{k2}(\bz_k)$ of line 11 and the functions $\theta(\bz_k)$ of line 14. Thus the algorithm does not stop until $r_{k^*}=r_{k^*-1}$. Then by $p\ge r_0>r_1>\cdots>r_{k^*-1}=r_{k^*}\ge 0$, it can be deduced that $0\le k^*\le p$.

(ii) $\Xi^*$ is index-0 because $E_{k^*}(\bz_{k^*},\delta)$ is of full row rank over $\mathcal K(\delta]$. Now consider the  $1,\dots,k$ steps of Algorithm \ref{Alg:2}, the unimodular matrice  $Q_k(\bz_k,\delta)$ for each $k$ does  not change solutions,  we have that $z_{k+1}(t)$ is a solution of $E_{k+1}(\bz_{k+1},\delta)\dot z_{k+1}=F_{k+1}(\bz_{k+1})$ if and only if $x(t)=z_0(t)$ is a solution of $\Xi$, where 
\begin{align}\label{Eq:maps}
    \begin{aligned}
    x(t)=z_0(t)&=\varphi^{-1}_0(\bz_1(t), 0), \\ &\cdots,\\ z_{k}(t)&=\varphi^{-1}_k(\bz_{k+1}(t), 0).
    \end{aligned}
\end{align}
Each bicausal change of coordinates $\varphi_k(z_k,\delta)$ is defined on $\mathcal K^{n_{k-1}}$, we extend it to $\mathcal K^n$ by setting  $\Phi_k=[\varphi_k,\bar z_{k},\bar z_{k-1}, \ldots,\bar z_1]^T$ (set $\Phi_0=\varphi_0$).  Notice that if $k^*=0$, then $\Xi^*$ coincides with $\Xi$ and item (ii) holds immediately. If $k^*>0$, assume without loss of generality that $k+1=k^*$, then $\Phi:=\Phi_{k^*}=[z_{k+1},\bar z_{k+1},\ldots,\bar z_1]^T=[z_{k^*},\bar z_{k^*},\ldots,\bar z_1]^T$ maps any solution $x(t)$ (and its delays) of $\Xi$  to $(z_{k+1}(t),0,\ldots,0)$ by (\ref{Eq:maps}), where $z_{k+1}(t)=z_{k^*}(t)$ is a solution of  $\Xi^*$.

(iii) Rewite $\Xi^*$ as
$$
E^0_{k^*}(\bz_{k^*}) \dot z_{k^*}=F_{k^*}(\bz_{k^*})-\sum\limits_{j=1}^{\bar j_{z_{k^*}}}E_{k^*}^j(\bz_{k^*})\delta^j\dot z_{k^*}.
$$
If $\rk_{\mathcal K}E^0_{k^*}(\bz_{k^*})=r_{k^*}$, then we can always find the right-inverse $(E^0_{k^*})^{\dagger}(\bz_{k^*})$ of $E^0_{k^*}(\bz_{k^*})$ over $\mathcal K$. Thus all solutions of $\Xi^*$ are solutions of the followings delayed ODE corresponding to all choices of free variables $v=v(t)$:
$$
\dot z_{k^*}=(E^0_{k^*})^{\dagger}F_{k^*}(\bz_{k^*})-\sum\limits_{j=1}^{\bar j}(E^0_{k^*})^{\dagger}E_{k^*}^j(\bz_{k^*})\delta^j\dot z_{k^*}+g(\bz_{k^*})v,
$$
where $g(\bz_{k^*})\in \mathcal K^{n_{k^*}\times (n_{k^*}-r_{k^*})}$ is of full column rank over $\mathcal K$ and $E^0_{k^*}(\bz_{k^*})g(\bz_{k^*})=0$.  
So $\Xi^*$   has a unique solution with an initial-value function $\xi_{z_{k^*}}$ if and only if the free variables $v$ is absent, i.e., $n_{k^*}-r_{k^*}=0$. Finally,  since $\Xi^*$ and $\Xi$ have isomorphic solutions by item (ii), we have that item (iii) holds.
\end{proof}
\begin{example}\label{Ex:DDAE}
Consider the following nonlinear DDAE $\Xi$, given by 
$$
\smatrix{1&0&0&0\\
0&1&0&0\\
0&0&1&0\\
0&0&0&0\\
0&0&0&0}\smatrix{\dot x_1\\ \dot x_2 \\ \dot x_3\\ \dot x_4}\!= \!\smatrix{x_2\\x_2^3x_1(-1)/(\ln c-x_1)\\ e^{x_1(-3)+x_3(-2)x_2(-3)}-c\\ x_1(-1)-x_1(-2)+x_3x_2(-1)-x_3(-1)x_2(-2)}.
$$
where $c>0$ is a constant. We apply Algorithm \ref{Alg:2} to $\Xi$. For $k=0$,  $E_0=E\in \mathbb R^{5\times 4}$ is constant and $\rk_{\mathcal K(\delta]} E_0=r_0=3$. Let $Q_0={\rm I}_5$, we get $F_{02}=\smatrix{e^{x_1(-3)+x_3(-2)x_2(-3)}+c\\x_1(-1)-x_1(-2)+x_3x_2(-1)-x_3(-1)x_2(-2)}$ (i.e., line 7). By a direct calculation, it is found that $\mathcal F_0=\kspan{\rd F_{02}}$ is closed and $\dim \kspan{\rd F_{02}}=1<p-r_0=2$. Thus we use the results of Proposition \ref{Pro:1} to find $\bar F_{02}=x_1(-1)+x_3x_2(-1)$ such that $\kspan{\rd \bar F_{02}}=\mathcal F_{0}$. It can be checked by applying Algorithm~\ref{Alg:1} to $\bar F_{02}$ that $\mathcal F_{0}$ satisfies \textbf{(C)}. We modify    $\bar F_{02}$ to $\tilde F_{02}=x_1(-1)+x_3x_2(-1)-\ln{c} $ such that $\tilde F_{02}=0$ is equivalent to $F_{02}=0$ (i.e., line 11). Then $\bar z_1=\tilde F_{02}=x_1(-1)+x_3x_2(-1)-\ln{c}$, $z_1=[x_1, x_2,x_4]^T$ is a bicausal change of coordinates and in  $(z_1, \bar z_1)$-coordinates, $\Xi$ becomes (i.e., line 15) $\smatrix{\tilde E_{01}(\bz_1,\bar \bz_1,\delta)&\tilde E_{02}(\bz_1,\bar \bz_1,\delta)\\ 0&0}\smatrix{\dot z_1\\ \dot {\bar z}_1}=\smatrix{F_{01}(\bz_1,\bar \bz_1)\\ F_{02}(\bz_1,\bar \bz_1)}:$
$$
\smatrix{1&0&0&0\\ 0&1&0&0 \\-\frac{1}{x_2(-1)}\delta&\frac{\bar z_1-\ln{c}+x_1(-1)}{x^2_2(-1)}\delta&0&-\frac{1}{x_2(-1)}\delta \\0&0&0&0\\0&0&0&0} \smatrix{\dot x_1\\ \dot x_2 \\ \dot x_4\\ \dot {\bar z}_1}\!=\!\smatrix{x_2\\ \frac{x_2^3x_1(-1)}{\ln c-x_1}\\ -x_4x_1(-1)\\ ce^{\bar z_1(-2)}-c  \\ \bar z_1-\bar z_1(-1)}.
$$
Thus by setting $\bar \bz_1=0$, we get (i.e., line 16)
$$
E_1=\tilde E_{01}= \smatrix{1&0&0\\ 0&1&0 \\ -\frac{1}{x_2(-1)}\delta&\frac{-\ln{c}+x_1(-1)}{x^2_2(-1)}\delta&0}, \ \ F_1=\smatrix{x_2\\ \frac{x_2^3x_1(-1)}{\ln{c}-x_1}\\ -x_4x_1(-1)}.
$$

Now set $k=1$ and go from line 17 to line 2. We have $\rk_{\mathcal K(\delta]} E_1=r_1=2<r_0$. Choose $Q_2(\bz_1,\delta)= \smatrix{1&0&0\\ 0&1&0 \\ \frac{1}{x_2(-1)}\delta&\frac{\ln{c}-x_1(-1)}{x^2_2(-1)}\delta&1}$ to define $F_{12}=1+x_2(-1)x_1(-2)-x_4x_1(-1)$. We can check via Algorithm \ref{Alg:1} that $\kspan{\rd F_{12}}$ satisfies \textbf{(C)} and $\bar z_2=F_{12}$, $z_2=[\tilde x_1,\tilde x_2]^T=[x_1,x_2x_1(-1)]^T$ define a bicausal change of $z_1$-coordinates. By  similar calculations as for $k=0$, we have that (line 16)
$$
E_2=\smatrix{1&0\\ \frac{\tilde x_2}{\tilde x_1(-1)}\delta&\tilde x_1(-1)}, \ \ F_2=\smatrix{\frac{\tilde x_2}{\tilde x_1(-1)}\\ \frac{\tilde x_2^3}{\tilde x^2_1(-1)(\ln{c}-\tilde x_1)}}.
$$
Go from line 16 to line 2 and set $k=2$, we have $\rk_{\mathcal K(\delta]} E_2=r_2=2=r_1$, thus Algorithm \ref{Alg:2} returns to $k^*=2$ and $z^*=z_2$. The DDAE $\Xi^*: E_2(\bz^*,\delta)\dot z^*=F_2(\bz^*)$ is clearly index-0 and we can rewrite it as a  delayed ODE of the neutral-type:
 \begin{align}\label{Eq:DODE}
 \dot {\tilde x}_1=\frac{\tilde x_2}{\tilde x_1(-1)},\ \ \dot {\tilde x}_2=\frac{\tilde x_2^3}{\tilde x^3_1(-1)(\ln{c}-\tilde x_1)}-\frac{\tilde x_2}{\tilde x^2_1(-1)}\dot {\tilde x}_1(-1). 
 \end{align}
 Given  initial-value conditions $\tilde x_1(s)=\xi_{\tilde x_1}(s)$, $s\in [-1,0]$ and $\tilde x_2(s)=\xi_{\tilde x_2}(s)$, $s\in [-1,0]$, we can calculate the solution $(\tilde x_1(t),\tilde x_2(t))$ of (\ref{Eq:DODE}) with respect to $(\xi_{\tilde x_1},\xi_{\tilde x_2})$ by the step method. Hence by Theorem \ref{Thm:DDAE} (ii), $\Phi^{-1}({\tilde x_1}(t),{\tilde x_2}(t),0,0)$ is the solution of $\Xi$ with the initial-value conditions $\Phi^{-1}(\xi_{\tilde x_1},\xi_{\tilde x_2},0,0)$, where $\Phi=[x_1,x_2x_1(-1),\bar z_2,\bar z_1]^T=[x_1,x_2x_1(-1),1+x_2(-1)x_1(-2)-x_4x_1(-1),x_1(-1)+x_3x_2(-1)-\ln{c}]^T$ is a bicausal change of coordinates. 
 \end{example}

\section{Conclusions and Perspectives}\label{sec:6}
 In order to generalize the implicit function theorem to the time-delay case,  we have proposed two extra equivalent conditions to the results of bicausal changes of coordinates in \cite{califano2016accessibility}. A technical lemma and an iterative algorithm are given to check those equivalent conditions. Moreover, we show that the generalized implicit function theorem can be used for reducing the index and for solving time-delay differential-algebraic equations.
 
 There are some further problems can be investigated based on our results. The example in Remark \ref{Rem:2} shows that it is possible to find a weaker condition for the time-delay implicit function theorem.    Another problem is to extend Algorithm \ref{Alg:2} to the general case when $\mathcal F_k$ does not satisfies \textbf{(C)}. Moreover, an interesting observation from Example \ref{Ex:DDAE}, which has already been pointed out in \cite{campbell19912,ascher1995numerical}, is that even the original DDAE $\Xi$ has a form of retarded type, the resulting delayed ODE can still be of neutral type (or even advanced type in general),  the problem of finding when a given DDAE can be reformulated as a delayed ODE of retarded, neutral, or advanced type, is open and challenging. Another   topic is to use our results to design reduce-order observers for both the states \cite{califano2020observability} and the inputs \cite{chen2022strong} of time-delay systems.

\bibliographystyle{ieeetrans}
\bibliography{bibchen}

\end{document}